\documentclass[review]{elsarticle}

\journal{Linear Algebra and its Applications}
\usepackage{times}

\usepackage{amssymb}
\usepackage{amsxtra}
\usepackage{amsmath}

\usepackage{verbatim}
\usepackage{epsfig}
\usepackage{setspace}
\usepackage{caption}
\usepackage{subcaption}
\usepackage{breqn}

\allowdisplaybreaks

\newtheorem{theorem}{Theorem}
\newtheorem{definition}{Definition}
\newtheorem{lemma}[theorem]{Lemma}
\newtheorem{proposition}[theorem]{Proposition}
\newtheorem{corollary}[theorem]{Corollary}
\newtheorem{example}{Example}[section]
\newtheorem{conjecture}{Conjecture}[section]
\newtheorem{remark}{Remark}
\newenvironment{proof}{\paragraph{Proof}}{\hfill$\square$}

\begin{document}

\begin{frontmatter}

\title{The Lattice Structure of Linear Subspace Codes\tnoteref{mytitlenote}}
\tnotetext[mytitlenote]{Content of this paper was partially covered in the doctoral thesis of the first author at Indian Institute of Science, Bangalore.}

\author{Pranab Basu\fnref{myfootnote}}
\ead{pbmobius@gmail.com}

\author{Navin Kashyap\fnref{myfootnote}}
\ead{nkashyap@iisc.ac.in}
\fntext[myfootnote]{The authors are with the Department of Electrical Communication Engineering, Indian Institute of Science, Bangalore.}



\begin{abstract}
The projective space $\mathbb{P}_q(n)$, i.e. the set of all subspaces of the vector space $\mathbb{F}_q^n$, is a metric space endowed with the subspace distance metric. Braun, Etzion and Vardy argued that codes in a projective space are analogous to binary block codes in $\mathbb{F}_2^n$ using a framework of lattices. They defined linear codes in $\mathbb{P}_q(n)$ by mimicking key features of linear codes in the Hamming space $\mathbb{F}_2^n$. In this paper, we prove that a linear code in a projective space forms a sublattice of the corresponding projective lattice if and only if the code is closed under intersection. The sublattice thus formed is geometric distributive. We also present an application of this lattice-theoretic characterization.
\end{abstract}

\begin{keyword}
linear codes, projective lattice, subspace codes.
\MSC[2010] 06B35
\end{keyword}

\end{frontmatter}


\section{Introduction}
\label{S1}

Let $\mathbb{F}_q$ be the unique finite field with $q$ elements, where $q$ is a prime power. The \emph{projective space} $\mathbb{P}_q(n)$ is the collection of all subspaces of $\mathbb{F}_q^n$, the finite vector space of dimension $n$ over $\mathbb{F}_q$. In terms of notation,
\begin{equation*}
\mathbb{P}_q(n) := \{X: X \le \mathbb{F}_q^n\},
\end{equation*}
where $\le$ denotes the usual vector space inclusion. The \emph{Grassmannian} of dimension $k$, denoted by $\mathbb{G}_q(n, k)$, for all nonnegative integers $k \le n$ is defined as the set of all $k$-dimensional subspaces of $\mathbb{F}_q^n$, i.e. $\mathbb{G}_q(n, k) := \{Y \in \mathbb{P}_q(n) : \dim Y = k\}$. Thus $\mathbb{P}_q(n) = \bigcup\limits_{k=0}^{n} \mathbb{G}_q(n, k)$. The \emph{subspace distance}, defined as
\begin{eqnarray*}
d_S(X, Y) &:=& \dim (X+Y) - \dim (X \cap Y) \nonumber \\
&=& \dim X + \dim Y - 2\dim (X \cap Y),
\end{eqnarray*}
for all $X, Y \in \mathbb{P}_q(n)$, is a metric for $\mathbb{P}_q(n)$ \cite{KK, AAK}, where $X+Y$ denotes the smallest subspace containing both $X$ and $Y$. This turns both $\mathbb{P}_q(n)$ and $\mathbb{G}_q(n, k)$ into metric spaces. An $(n, M, d)$ code $\mathbb{C}$ in $\mathbb{P}_q(n)$ is a subset of the projective space with size $|\mathbb{C}| = M$ such that $d_S(X, Y) \ge d$ for all $X, Y \in \mathbb{C}$. The parameters $n$ and $d$ are called the \emph{length} and \emph{minimum distance} of the code, respectively. A code in a projective space is commonly referred to as a \emph{subspace code}. A subspace code $\mathbb{C}$ is called a \emph{constant dimension code} with fixed dimension $k$ if $\dim X = k$ for all $X \in \mathbb{C}$. In other words, $\mathbb{C} \subseteq \mathbb{G}_q(n, k)$ for some $k \le n$ if $\mathbb{C}$ is a constant dimension code. Koetter and Kschischang proved that in random network coding, a subspace code with minimum distance $d$ can correct any combination of $t$ errors and $\rho$ erasures introduced anywhere in the network if $2(t + \rho) < d$ \cite{KK}. This development triggered interest in codes in projective spaces in recent times \cite{EV, SE, SE2, HKK, GY, XF, TMBR, KoK, ER, BP, GR}.

We denote the collection of all subsets of the canonical $n$-set $[n] := \{1, \ldots, n\}$ as $\mathcal{P}(n)$, commonly known as the \emph{power set} of $[n]$. The authors of \cite{BEV} proved that codes in projective spaces can be viewed as the $q$-analog of binary block codes in Hamming spaces using the framework of lattices. A \emph{lattice} is a partially ordered set wherein any two elements have a least upper bound and a greatest lower bound existing within the set. Block codes in $\mathbb{F}_2^n$ correspond to the \emph{power set lattice} $(\mathcal{P}(n), \cup, \cap, \subseteq)$ while subspace codes in $\mathbb{P}_q(n)$ correspond to the \emph{projective lattice} $(\mathbb{P}_q(n), +, \cap, \le)$. Here $\subseteq$ signifies set inclusion. Braun et al. generalized a few properties of binary block codes to subspace codes including that of linearity \cite{BEV}. Linear codes in Hamming spaces find huge application in designing error correcting codes due to their structure \cite{MS, DR}. An $[n, k]$ linear block code in $\mathbb{F}_q^n$ is precisely a $k$-dimensional subspace of $\mathbb{F}_q^n$, where $k$ is the \emph{dimension} of such a code.

The notion of ``linearity'' in a projective space, however, is not straightforward. This stems from projective spaces not exhibiting vector space structure unlike Hamming spaces. In particular, $\mathbb{F}_q^n$ is a vector space with respect to the bitwise XOR-operation whereas $\mathbb{P}_q(n)$ is not a vector space under the usual vector space addition. Braun et al. solved this problem in \cite{BEV} by assigning a vector space-like structure to a subset of $\mathbb{P}_q(n)$.


The \emph{rate} of a linear code, i.e. the ratio of its dimension to length, is proportional to the size of the code. It is natural to ask how large a linear code in $\mathbb{P}_q(n)$ can be. Braun et al. conjectured the following in \cite{BEV}:

\begin{conjecture}
	\label{C}
	The maximum size of a linear code in $\mathbb{P}_q(n)$ is $2^n$.
	\end{conjecture}

Special cases of Conjecture~\ref{C} have been proved before \cite{BK, PS}. We proved the conjecture in \cite{BK} under the additional assumption of the codes being closed under intersection. In this paper, we bring out the lattice structure of linear subspace codes closed under intersection. In particular, we show that linear subspace codes are sublattices of the projective lattice if and only if they are closed under intersection. Moreover, these sublattices are geometric distributive. We then go on to use the lattice-theoretic characterization of this particular class of linear codes to give an alternative proof of the conjectured bound for them.
  


The rest of the paper is organized as follows. In Section~\ref{S2} we give the formal definition of a linear code in a projective space and some relevant definitions from lattice theory. Several properties of linear subspace codes are derived that highlight the $q$-analog structure of a binary linear block code. The Union-Intersection theorem is stated and proved in Section~\ref{S3}. As a consequence, we show the lattice structure of linear codes closed under intersection. We introduce the notion of pairwise disjoint codewords in linear subspace codes and establish some properties to show their linear independence in Section~\ref{S4}. Section~\ref{S5} is devoted to proving that the sublattice of the projective lattice formed by a linear code closed under intersection is geometric distributive. The proof uses the notion of indecomposable codewords which are particular cases of pairwise disjoint codewords. We then use the lattice-theoretic characterization to give an alternative proof of the maximal size of linear codes closed under intersection. Section~\ref{S6} contains a few open problems for future research.

\paragraph{Notation}
$\mathbb{F}_q^n$ denotes the finite vector space of dimension $n$ over $\mathbb{F}_q$. The set of all subspaces of $\mathbb{F}_q^n$ is denoted by $\mathbb{P}_q(n)$. The usual vector space sum of two subspaces $X$ and $Y$ when $X \cap Y = \{0\}$, also known as the \emph{direct sum} of $X$ and $Y$, will be written as $X \oplus Y$. For a binary vector or \emph{word} $x = (x_1, \ldots, x_n) \in \mathbb{F}_2^n$ of length $n$, the \emph{support} of $x$, denoted as $supp(x)$, will indicate the set of nonzero coordinates of $x$. In other words, $supp(x) := \{i: i \in [n], x_i = 1\}$. The support of a binary vector identifies it completely. The all-zero vector and the empty set will be denoted as $\mathbf{0} := (0, \ldots, 0)$ and $\emptyset$, respectively. For two binary words $x$ and $y$, the \emph{union} and \emph{intersection} of $x$ and $y$, denoted as $x \circ y$ and $x \ast y$ respectively, are defined via
\begin{eqnarray*}
supp(x \circ y) = supp(x) \cup supp(y); \\
supp(x \ast y) = supp(x) \cap supp(y).
\end{eqnarray*}
The coordinatewise modulo-2 addition, alternatively known as the \emph{binary vector addition}, of two binary words $x$ and $y$ is denoted by $x + y$. By definition, $supp(x + y) = supp(x) \triangle supp(y)$. Here $\triangle$ denotes the \emph{symmetric difference} operator, defined for sets $A$ and $B$ as
\begin{equation*}
	A \triangle B := (A \cup B) \backslash (A \cap B).
\end{equation*}
\section{Definitions and Relevant Background}
\label{S2}
\subsection{Linear Codes in Projective Spaces}

A linear code $\mathcal{U}$ in the projective space $\mathbb{P}_q(n)$ is defined as follows \cite{BEV}:
\begin{definition}
	\label{L}
	A subset $\mathcal{U}  \subseteq \mathbb{P}_q(n)$, with $\left\{0\right\} \in \mathcal{U}$, is a \emph{linear subspace code} if there exists a function $\boxplus : \mathcal{U} \times \mathcal{U} \rightarrow \mathcal{U}$ such that: \\
	(i) $(\mathcal{U}, \boxplus)$ is an abelian group; \\
	(ii) the identity element of $(\mathcal{U}, \boxplus)$ is $\left\{0\right\}$; \\
	(iii) $X \boxplus X = \left\{0\right\}$ for every group element $X \in \mathcal{U}$; \\
	(iv) the addition operation $\boxplus$ is isometric, i.e., $d_S(X \boxplus Y_1, X \boxplus Y_2) = d_S(Y_1, Y_2)$ for all $X, Y_1, Y_2 \in \mathcal{U}.$
\end{definition}
A subset $\mathcal{U}$ of $\mathbb{P}_q(n)$, together with corresponding $\boxplus$ operation, is called a \emph{quasi-linear code} if it satisfies only the first three conditions in the above definition. Conditions (i)-(iii) in Definition \ref{L} ensure that a quasi-linear code is a vector space over $\mathbb{F}_2$. Braun et al. proved the following about the size of a quasi-linear code in $\mathbb{P}_q(n)$ \cite{BEV}.
\begin{proposition}(\cite{BEV}, Proposition 2)
	\label{1}
	A subset $\mathcal{U} \subseteq \mathbb{P}_q(n)$, with $\left\{0\right\} \in \mathcal{U}$, is a quasi-linear code if and only if $\left|\mathcal{U}\right|$ is a power of 2.
\end{proposition}
A set of codewords in a linear subspace code will be said to be \emph{linearly independent} if the members of the set are linearly independent vectors in the vector space formed by the code over $\mathbb{F}_2$.

A quasi-linear code is linear when translation invariance is imposed on its structure, as indicated by condition (iv) in Definition \ref{L}. The linear addition $\boxplus$ thus becomes isometric and obeys certain properties. We list and prove these as lemmas, the first three of which are essentially reproduced from \cite{BEV}.
\begin{lemma}(\cite{BEV}, Lemma 6)
	\label{2}
	Let $\mathcal{U}$ be a linear code in $\mathbb{P}_q(n)$ and let $\boxplus$ be the isometric linear addition on $\mathcal{U}$. Then for all $X, Y \in \mathcal{U}$, we have:
	\begin{equation*}
	\dim(X \boxplus Y) = d_S(X, Y) = \dim X + \dim Y - 2\dim(X \cap Y)
	\end{equation*}
	In particular, if $X \subseteq Y$, then $\dim(X \boxplus Y) = \dim Y - \dim X$.
\end{lemma}
\begin{proof}
	By the definition of linear code, $d_S(X, Y) = d_S(X \boxplus Y, Y \boxplus Y) = d_S(X \boxplus Y, \left\{0\right\}) = \dim(X \boxplus Y)$. From the definition of $d_S(X, Y)$ the result follows.
\end{proof}

\begin{lemma}(\cite{BEV}, Lemma 7)
	\label{3}
	For any three subspaces $X, Y$ and $Z$ of a linear code $\mathcal{U}$ in $\mathbb{P}_q(n)$ with isometric linear addition $\boxplus$, the condition $Z = X \boxplus Y$ implies $Y = X \boxplus Z$.
\end{lemma}
\begin{proof}
	From the definition of linearity in $\mathbb{P}_q(n)$, we have $Y = (X \boxplus X) \boxplus Y = X \boxplus (X \boxplus Y) = X \boxplus Z$.
\end{proof}

The statement of the next lemma is altered from what was presented in \cite{BEV} as per our requirement.
\begin{lemma}(\cite{BEV}, Lemma 8)
	\label{4}
	Let $\mathcal{U}$ be a linear code in $\mathbb{P}_q(n)$ and let $\boxplus$ be the isometric linear addition on $\mathcal{U}$. If $X$ and $Y$ are any two codewords of $\mathcal{U}$ such that $X \cap Y = \left\{0\right\}$, then $X \boxplus Y = X \oplus Y$. Also $\dim(X \boxplus Y) = \dim X + \dim Y$.
\end{lemma}
\begin{proof}
	From the definition of linearity, we have $\dim X = \dim((X \boxplus Y) \boxplus Y) = \dim(X \boxplus Y) + \dim Y - 2\dim((X \boxplus Y) \cap Y)$ and using the fact $X \cap Y = \left\{0\right\}$, we also have from Lemma~\ref{2} that $\dim(X \boxplus Y) = \dim X + \dim Y$. Combining both, we obtain $\dim X = \dim X + 2\dim Y - 2\dim((X \boxplus Y) \cap Y)$, which implies $\dim Y = \dim((X \boxplus Y) \cap Y)$, i.e., $Y \subseteq (X \boxplus Y)$. Similarly, $X \subseteq (X \boxplus Y)$. This means $X + Y \subseteq X \boxplus Y$. Finally, as $X \cap Y = \left\{0\right\}$, $\dim(X + Y) = \dim X + \dim Y = \dim(X \boxplus Y)$, which proves the lemma.
\end{proof}

The next lemma, which plays a pivotal role in our work, records some useful properties of the dimension of codewords in a linear subspace code.
\begin{lemma}
	\label{6d}
	If $\mathcal{U}$ is a linear subspace code and $X, Y \in \mathcal{U}$, then
	\begin{itemize}
		\item[(i)] $\dim X = \dim (X \cap Y) + \dim(X \cap (X \boxplus Y))$,
		\item[(ii)] $\dim(X \boxplus Y) = \dim(X \cap (X \boxplus Y)) + \dim(Y \cap (X \boxplus Y))$.
	\end{itemize}
\end{lemma}
\begin{proof}
	(i) By Definition~\ref{L}, every element of $\mathcal{U}$ is a self-inverse and $(\mathcal{U}, \boxplus)$ is an abelian group, hence $Y = (X \boxplus Y) \boxplus X$, and by applying Lemma~\ref{2} we get,
	\begin{equation*}
	\dim Y = \dim(X \boxplus Y) + \dim X - 2 \dim(X \cap (X \boxplus Y)).
	\end{equation*}
	Expanding $\dim(X \boxplus Y)$ using Lemma~\ref{2} and cancelling like terms from both sides, we get the desired result. \\
	(ii) Follows from (i) after substituting $X$ with $X \boxplus Y$. By Lemma~\ref{3}, $X \boxplus Y$ is then replaced by $(X \boxplus Y) \boxplus Y = X$, and the result follows.
\end{proof}

The dimension of the $\boxplus$ sum of two codewords in a linear subspace code is bounded from below, as shown next.
\begin{lemma}
	\label{6e}
	Let $\mathcal{U}$ be a linear subspace code. For all $X, Y \in \mathcal{U}$ the following is true:
	\begin{equation*}
	\dim (X \boxplus Y) \geq \dim X - \dim Y,
	\end{equation*}
	with equality if and only if $Y \subseteq X$.
\end{lemma}
\begin{proof}
	As $X \cap Y \subseteq Y$, we must have $\dim (X \cap Y) \leq \dim Y$ and equality occurs only when $X \cap Y = Y$, i.e., when $Y \subseteq X$. Lemma~\ref{2} then implies that
	\begin{eqnarray*}
		\dim (X \boxplus Y) &=& \dim X + \dim Y - 2\dim (X \cap Y) \\
		&\geq& \dim X + \dim Y - 2\dim Y \\
		&=& \dim X - \dim Y.
	\end{eqnarray*}
	Equality occurs if and only if $\dim(X \cap Y) = \dim Y$, i.e. if and only if $Y \subseteq X$. 
\end{proof}
\begin{lemma}
	\label{6i}
	Let $\mathcal{U}$ be a linear subspace code and $X, Y$ be two distinct nontrivial codewords of $\mathcal{U}$. Then,
	\begin{itemize}
		\item[(a)] $Y \subset X$ if and only if $(X \boxplus Y) \subset X$.
		\item[(b)] $Y \subset X$ if and only if $Y \cap (X \boxplus Y) = \{0\}$.
	\end{itemize}
\end{lemma}
\begin{proof}
	(a)
	(Proof of $\Rightarrow$) By Lemma~\ref{2}, $\dim (X \boxplus Y) = \dim X + \dim Y - 2\dim (X \cap Y)$. If $Y \subset X$, then $X \cap Y = Y$ and we have, $\dim (X \boxplus Y) = \dim X + \dim Y - 2\dim Y = \dim X - \dim Y$. On the other hand, by Lemma~\ref{6d} and the fact that $Y \subset X$, we have $\dim (X \cap (X \boxplus Y)) = \dim X - \dim (X \cap Y) = \dim X - \dim Y$. Thus $\dim (X \cap (X \boxplus Y)) = \dim (X \boxplus Y)$, which proves that $(X \boxplus Y) \subseteq X$. However, if $X \boxplus Y = X$, then by Definition~\ref{L}, $Y = \{ 0\}$, a contradiction, hence proved.  \\
	(Proof of $\Leftarrow$) Write $Z = X \boxplus Y$. If $Z \subset X$, then by logic similar to that presented above, $(X \boxplus Z) \subset X$. As $Y = X \boxplus Z$, the result follows. \\
	(b)
	(Proof of $\Rightarrow$) By Definition~\ref{L}, $X = (X \boxplus Y) \boxplus Y$, and using Lemma~\ref{2} we get
	\begin{equation}
	\label{E5a}
	\dim X = \dim(X \boxplus Y) + \dim Y - 2\dim((X \boxplus Y) \cap Y).
	\end{equation}
	If $Y \subset X$, by Lemma~\ref{6e} we have $\dim(X \boxplus Y) = \dim X - \dim Y$, which reduces \eqref{E5a} to $\dim ((X \boxplus Y) \cap Y) = 0$, whence the result follows. \\
	(Proof of $\Leftarrow$) Suppose $(X \boxplus Y) \cap Y = \{0\}$. According to Lemma~\ref{4} and Definition~\ref{L}, we get $\dim X = \dim((X \boxplus Y) \boxplus Y) = \dim (X \boxplus Y) + \dim Y$, whence $\dim (X \boxplus Y) = \dim X - \dim Y$, and the result follows by Lemma~\ref{6e}.
\end{proof}
\begin{remark}
	Equivalent results of Lemmas~\ref{2}--\ref{6i} for linear codes in $\mathbb{F}_2^n$ have already been established \cite{MS} or can easily be deduced.
\end{remark}
\subsection{An Overview of Lattices}

This section serves as a brief introduction to lattices. We will give a few basic definitions that can be found in \cite{B}. The notation and terminology used here are standard.

\begin{definition}
	\label{D1}
	A partially ordered set or poset $(P, \preceq)$ is a set $P$ in which a binary relation $\preceq$ is defined which satisfies the following conditions for all $x, y, z \in P$:
	\begin{itemize}
		\item[\textbf{$P_1$}.] (Reflexivity) For all $x$, $x \preceq x$.
		\item[\textbf{$P_2$}.] (Antisymmetry) If $x \preceq y, y \preceq x$, then $x = y$.
		\item[\textbf{$P_3$}.] (Transitivity) If $x \preceq y, y \preceq z$, then $x \preceq z$.
	\end{itemize}
\end{definition}

The binary relation $\preceq$ in a poset $(P, \preceq)$ is also called the \emph{order relation} for the poset. We will henceforth denote a poset $(P, \preceq)$ by $P$ and assume $\preceq$ as its order relation. If $x \preceq y$ and $x \neq y$, we will write $x \prec y$ and say that $x$ is ``less than'' or ``properly contained in'' $y$. If $x \prec y$ and there exists no $z \in P$ such that $x \prec z \prec y$, then $y$ is said to \emph{cover} the element $x$; we denote this as $x \lessdot y$.
\begin{definition}
	\label{D2}
	An upper bound of a subset $S$ of $P$ is an element $a \in P$ such that $s \preceq a$ for all $s \in S$. Similarly, the lower bound of a subset $S$ of $P$ is an element $b \in P$ satisfying $b \preceq s$ for every $s \in S$. The least upper bound (greatest lower bound) of $S$ is the element of $P$ contained in (containing) every upper bound (lower bound) of $S$.
\end{definition}

A least upper bound of a poset, if it exists, is unique according to the antisymmetry property of the order relation (\textbf{$P_2$}, Definition~\ref{D1}). Same holds for a greatest lower bound of a poset. We will use notations $\sup S$ and $\inf S$ for the least upper bound and greatest lower bound of a poset $S$, respectively.
\begin{definition}
	\label{D3}
	A lattice is a poset $L$ such that for any $x, y \in L$, $\inf \{ x, y\}$ and $\sup \{x, y\}$ exist. The $\inf \{ x, y\}$ is denoted as $x \wedge y$ and read as ``$x$ meet $y$'', while the $\sup \{x, y\}$ is denoted as $x \vee y$ and read as ``$x$ join $y$''. The lattice is denoted as $(L, \vee, \wedge)$. The unique least upper bound (greatest lower bound) of the whole lattice $L$, if it exists, is called the greatest (least) element of $L$.
\end{definition}

All the lattices considered in this work are finite and contain a unique greatest element denoted as $I$ and a unique least element denoted as $O$.
\begin{definition}
	\label{D4}
	A sublattice of a lattice $L$ is a subset $X$ of $L$ such that for all $a, b \in X$ it follows that $a \vee b \in X, a \wedge b \in X$.
\end{definition}

A sublattice is a lattice in its own right with the same meet and join operations as that of the lattice. However, not all subsets of a lattice are sublattices.
\begin{definition}
	\label{D5}
	A lattice $(L, \vee, \wedge)$ is distributive if any of the following two equivalent conditions holds for all $x, y, z \in L$:
	\begin{eqnarray*}
		x \vee (y \wedge z) = (x \vee y) \wedge (x \vee z) \\
		x \wedge (y \vee z) = (x \wedge y) \vee (x \wedge z).
	\end{eqnarray*}
\end{definition}
\begin{definition}
	\label{D6}
	A lattice $(L, \vee, \wedge)$ is modular if for all $a, b, c \in L$ such that $a \le c$, we have $a \vee (b \wedge c) = (a \vee b) \wedge c$.
\end{definition}

Not all lattices are distributive. If a lattice is distributive then the modularity condition automatically holds. Thus \emph{all distributive lattices are modular}. However, the opposite is not true as will be illustrated later.
In a lattice $L$, an element $x \in L$ is called an \emph{atom} if and only if $O \lessdot x$. Atoms play a significant role in defining lattices that are geometric.
\begin{definition}
	\label{D7}
	A finite lattice is geometric if it is modular and every element in the lattice is a join of atoms. If a geometric lattice is distributive then it is called geometric distributive.
\end{definition}

A set of elements $\{x_0, x_1, \ldots, x_n\}$ in a lattice is called a \emph{chain} if $x_i < x_{i+1}$ for all $0 \le i \le n-1$. The \emph{length} of this chain is $n$. The \emph{height} of a geometric lattice is the length of a maximal chain between its greatest and least elements.

Not all modular or distributive lattices are geometric. We will next discuss an example of a geometric lattice that is not distributive.
\begin{example}
	\label{PL}
	Recall that the projective space $\mathbb{P}_q(n)$ represents the set of all subspaces of $\mathbb{F}_q^n$, the finite vector space of dimension $n$ over $\mathbb{F}_q$. It is straightforward to verify that $(\mathbb{P}_q(n), \le)$ is a poset where the order relation is the usual subspace inclusion $\le$. The entire projective space is a lattice under this order relation. The join of two elements $X$ and $Y$ is therefore the smallest subspace containing both $X$ and $Y$. Similarly the meet of $X$ and $Y$ becomes the largest subspace contained in both $X$ and $Y$. Thus, in this lattice, the meet and join operations are defined as: $X \vee Y = X+Y, X \wedge Y = X \cap Y$ for all $X, Y \in \mathbb{P}_q(n)$. The greatest and least elements for this lattice are the ambient space $\mathbb{F}_q^n$ and the null space $\{0\}$, respectively. The atoms in $\mathbb{P}_q(n)$ are precisely the one dimensional vector spaces of $\mathbb{F}_q^n$, i.e., the set of atoms is $\mathcal{G}_q(n, 1)$. As $A + (B \cap C) = (A+B) \cap C$ for all $A, B, C \in \mathbb{P}_q(n)$ such that $A \subseteq C$, this lattice is modular. This, together with the fact that any element in the projective space is a union (vector space sum) of one dimensional subspaces, implies that the lattice is geometric. However, we do not have $(A \cap C) + (B \cap C) = (A+B) \cap C$ for all subspaces $A, B, C$ of $\mathbb{F}_q^n$ in general. Thus, the lattice is not distributive.
\end{example}

We refer to the lattice $(\mathbb{P}_q(n), +, \cap)$ as \emph{projective lattice}. Recall that any linear subspace code $\mathcal{U}$ in $\mathbb{P}_q(n)$ is a subset of $\mathbb{P}_q(n)$--- which, according to Definition~\ref{D4}, is not sufficient to guarantee a lattice structure. It is therefore natural to ask what additional condition(s) a linear code in a projective space should satisfy in order to assume a sublattice structure of the corresponding projective lattice. We investigate this problem in the following section.

\section{The Union-Intersection Theorem}
\label{S3}

We introduced the terms \emph{union} and \emph{intersection} of two codewords in a Hamming space in Section~\ref{S1}. The corresponding notions for linear codes in a projective space is straightforward: The union of two codewords $X$ and $Y$ is $X+Y$, while their intersection is $X \cap Y$. Observe that for any two codewords $x$ and $y$ in a linear code $\mathcal{C} \subseteq \mathbb{F}_2^n$, $supp((x \ast y) + (x \circ y)) = supp(x \ast y) \triangle supp(x \circ y) = (supp(x) \cap supp(y)) \triangle (supp(x) \cup supp(y)) = supp(x) \triangle supp(y) = supp(x+y)$, which proves that
\begin{equation}
\label{E12a}
x + y = (x \ast y) + (x \circ y).
\end{equation}
Thus the union and intersection of any two codewords in a classical binary linear code must coexist within the code according to \eqref{E12a}. Moreover, $supp((x \ast y) \ast (x + y)) = (supp(x) \cap supp(y)) \cap (supp(x) \triangle supp(y)) = \emptyset$, i.e. $(x \ast y) \ast (x+y) = \mathbf{0}$. We now prove that equivalent relations hold for linear codes in a projective space.
\begin{theorem}[Union-Intersection Theorem]
	\label{6p}
	Let $\mathcal{U}$ be a linear subspace code. If $X$ and $Y$ are two codewords in $\mathcal{U}$ then,
	\begin{equation*}
	X \cap Y \in \mathcal{U} \iff X + Y \in \mathcal{U}.
	\end{equation*}
	Furthermore, if $X \cap Y, X + Y \in \mathcal{U}$ then $(X \cap Y) \cap (X \boxplus Y) = \{0\}$, and
	\begin{equation*}
	X + Y = (X \boxplus Y) \boxplus (X \cap Y) = (X \boxplus Y) \oplus (X \cap Y).
	\end{equation*}
\end{theorem}
\begin{proof}
	(Proof of $\Rightarrow$) Assume that $X \cap Y \in \mathcal{U}$ for some $X, Y \in \mathcal{U}$. Since $X \cap Y \subseteq X$ and $X \cap Y \subseteq Y$, by Lemma~\ref{6e} we get:
	\begin{eqnarray*}
		\dim (X \boxplus (X \cap Y)) &=& \dim X - \dim(X \cap Y), \\
		\dim (Y \boxplus (X \cap Y)) &=& \dim Y - \dim(X \cap Y).
	\end{eqnarray*}
	
	We first prove that $X \cap (Y \boxplus (X \cap Y)) = Y \cap (X \boxplus (X \cap Y)) = \{0\}$. Having proved this, we will show that $X \cap Y \cap (X \boxplus Y) = \{0\}$, which will help us to establish that $X + Y = X \boxplus Y \boxplus (X \cap Y)$. As $X \boxplus Y, X \cap Y \in \mathcal{U}$, this will suffice to prove that $X+Y \in \mathcal{U}$.
	
	Suppose $Z = X \cap (Y \boxplus (X \cap Y))$. Thus $Z \subseteq Y \boxplus (X \cap Y)$, and as $X \cap Y \subseteq Y$, by Lemma~\ref{6i}(a) we have
	\begin{equation}
	\label{E15}
	Y \boxplus (X \cap Y) \subseteq Y.
	\end{equation}
	Thus, \eqref{E15} implies that $Z \subseteq Y$. Combining this with the fact that $Z \subseteq X$ (Since $Z = X \cap (Y \boxplus (X \cap Y))$), we obtain $Z \subseteq X \cap Y$. Therefore,
	\begin{equation}
	\label{E15a}
	Z \subseteq (X \cap Y) \cap (Y \boxplus (X \cap Y)).
	\end{equation}
	However, according to Lemma~\ref{6i}(b), $(X \cap Y) \cap (Y \boxplus (X \cap Y)) = \{0\}$, hence \eqref{E15a} implies that $Z = \{0\}$. Similarly we can prove that $Y \cap (X \boxplus (X \cap Y)) = \{0\}$. This establishes our first claim. 
	
	As $X \cap (Y \boxplus (X \cap Y)) = \{0\}$, by Lemma~\ref{4} and Definition~\ref{L} we can write:
	\begin{eqnarray}
	\label{E16}
	\dim(X \boxplus Y \boxplus (X \cap Y)) &=& \dim (X \boxplus (Y \boxplus (X \cap Y))) \nonumber \\
	&=& \dim X + \dim(Y \boxplus (X \cap Y)) \nonumber \\
	&=& \dim X + (\dim Y - \dim(X \cap Y)) \nonumber \\
	&=& \dim(X+Y).
	\end{eqnarray}
	Observe that $\dim(X \boxplus Y) + \dim(X \cap Y) = \dim X + \dim Y - \dim(X \cap Y) = \dim(X+Y)$. We can now calculate $\dim(X \boxplus Y \boxplus (X \cap Y))$ in a different way:
	\begin{eqnarray}
	\label{E17}
	\dim(X \boxplus Y \boxplus (X \cap Y)) &=& \dim((X \boxplus Y) \boxplus (X \cap Y)) \nonumber \\
	&=& \dim(X+Y) - 2\dim(X \cap Y \cap (X \boxplus Y)).
	\end{eqnarray}
	
	Combining \eqref{E16} and \eqref{E17} gives us: $\dim(X \cap Y \cap (X \boxplus Y)) = 0$, i.e.,
	\begin{equation}
	\label{E18}
	X \cap Y \cap (X \boxplus Y) = \{0\}.
	\end{equation}
	Since $X \cap (Y \boxplus (X \cap Y)) = \{0\}$, by Definition~\ref{L} and Lemma~\ref{4} we can express $X \boxplus Y \boxplus (X \cap Y)$ as: $X \boxplus Y \boxplus (X \cap Y) = X \boxplus (Y \boxplus (X \cap Y)) = X \oplus (Y \boxplus (X \cap Y))$, which indicates that $X \subseteq X \boxplus Y \boxplus (X \cap Y)$. In a similar fashion we can prove that $Y \subseteq X \boxplus Y \boxplus (X \cap Y)$, thereby enabling ourselves to write:
	\begin{equation}
	\label{E19}
	X + Y \subseteq X \boxplus Y \boxplus (X \cap Y).
	\end{equation}
	\eqref{E16} and \eqref{E19} together imply that $X+Y = X \boxplus Y \boxplus (X \cap Y)$, which establishes our final claim. By virtue of \eqref{E18} and Lemma~\ref{4}, we can also write: $X+Y = (X \boxplus Y) \oplus (X \cap Y)$. \\
	(Proof of $\Leftarrow$) We assume that $X+Y \in \mathcal{U}$ for some codewords $X$ and $Y$ in $\mathcal{U}$. Let us consider $W \in \mathcal{U}$ such that,
	\begin{equation}
	\label{E21}
	W = X \boxplus Y \boxplus (X+Y).
	\end{equation}
	We claim that $W = X \cap Y$, establishing which will suffice to prove that $X \cap Y \in \mathcal{U}$. Observe from \eqref{E21} that $W \boxplus X = Y \boxplus (X+Y), W \boxplus Y = X \boxplus (X+Y)$. Since $X \subseteq X+Y, Y \subseteq X+Y$, Lemma~\ref{6e} implies the following:
	\begin{eqnarray}
	\label{E22}
	\dim(W \boxplus X) &=& \dim(X+Y) - \dim Y = \dim X - \dim (X \cap Y), \nonumber \\
	\dim(W \boxplus Y) &=& \dim(X+Y) - \dim X = \dim Y - \dim (X \cap Y).
	\end{eqnarray}
	Since $(\mathcal{U}, \boxplus)$ is an abelian group wherein any element is self-inverse, we can express $X \boxplus Y$ as: $X \boxplus Y = (W \boxplus X) \boxplus (W \boxplus Y)$. Then applying Lemma~\ref{2} gives us:
	\begin{eqnarray}
	\begin{aligned}
	\dim(X \boxplus Y) = &\dim(W \boxplus X) + \dim(W \boxplus Y) - 2\dim((W \boxplus X) \cap (W \boxplus Y)) \nonumber \\
	= &(\dim X - \dim(X \cap Y)) + (\dim Y - \dim(X \cap Y)) \nonumber \\
	&- 2\dim((W \boxplus X) \cap (W \boxplus Y)) \nonumber \\
	= &\dim(X \boxplus Y) - 2\dim((W \boxplus X) \cap (W \boxplus Y)).
	\end{aligned}
	\end{eqnarray}
	The above expression clearly indicates that,
	\begin{equation}
	\label{E23}
	(W \boxplus X) \cap (W \boxplus Y) = \{0\}.
	\end{equation}
	Since $X \boxplus Y = (W \boxplus X) \boxplus (W \boxplus Y)$, an immediate consequence of \eqref{E23} after using Lemma~\ref{4} is that,
	\begin{equation}
	\label{E24}
	X \boxplus Y = (W \boxplus X) + (W \boxplus Y).
	\end{equation}
	As $Y \subseteq X+Y$, Lemma~\ref{6i}(a) implies that $W \boxplus X = Y \boxplus (X+Y) \subseteq X+Y$, Similarly, $X \subseteq X+Y$ implies that $W \boxplus Y \subseteq X+Y$. Therefore, \eqref{E24} yields that,
	\begin{equation}
	\label{E25}
	X \boxplus Y \subseteq X+Y.
	\end{equation}
	Applying Lemma~\ref{6e} after combining \eqref{E21} and \eqref{E25} results in the following:
	\begin{eqnarray}
	\label{E26}
	\dim W &=& \dim((X \boxplus Y) \boxplus (X+Y)) \nonumber \\
	&=& \dim(X+Y) - \dim(X \boxplus Y) \nonumber \\
	&=& \dim(X \cap Y).
	\end{eqnarray}
	We now compute $\dim(W \boxplus X)$ in two different ways: first by recalling \eqref{E22} and then by using Lemma~\ref{2}. Equating both the expressions, we get $\dim X - \dim(X \cap Y) = \dim W + \dim X - 2\dim(W \cap X)$, which, after applying \eqref{E26} and cancelling like terms, reduces to:
	\begin{equation}
	\label{E27}
	\dim W = \dim(W \cap X).
	\end{equation}
	\eqref{E27} implies that $W \subseteq X$. Using similar technique we can also obtain $W \subseteq Y$, which therefore gives us:
	\begin{equation}
	\label{E28}
	W \subseteq X \cap Y.
	\end{equation}
	Comparison of \eqref{E26} and \eqref{E28} then yields $W = X \cap Y$, which establishes our claim. To complete the proof, observe that $X+Y = W \boxplus (X \boxplus Y) = (X \cap Y) \boxplus (X \boxplus Y)$ follows from \eqref{E21} and Lemma~\ref{3}. We also obtain $(X \boxplus Y) \cap (X \cap Y) = (X \boxplus Y) \cap (X \boxplus Y \boxplus (X+Y)) = \{0\}$ using Lemma~\ref{6i}(b) as $X \boxplus Y \subseteq X+Y$. Thus, by Lemma~\ref{4},
	\begin{equation*}
	X+Y = (X \boxplus Y) \oplus (X \cap Y).
	\end{equation*}
\end{proof}
\begin{remark}
	We proved that $(X \cap Y) \cap (X \boxplus Y) = \{0\}$ when $X \cap Y, X + Y \in \mathcal{U}$ for a linear code $\mathcal{U} \subseteq \mathbb{P}_q(n)$. However, this is not necessarily true when $X \cap Y$ and $X + Y$ do not belong to the code. For example, consider a code $\mathcal{C} = \{\{0\}, X_1, X_2, X_3\} \subseteq \mathbb{P}_2(n)$, where $X_1, X_2, X_3 \in \mathcal{G}_2(n, 2)$ such that $X_1, X_2, X_3$ are distinct and $Z := X_1 \cap X_2 \cap X_3 \in \mathcal{G}_2(n, 1)$. Define a commutative function $\boxplus: \mathcal{C} \times \mathcal{C} \rightarrow \mathcal{C}$ as follows: $Y \boxplus Y := \{0\}$ and $Y \boxplus \{0\} := Y$ for all $Y \in \mathcal{C}$, while $X_i \boxplus X_j := X_k$ for any distinct $i, j, k \in [3]$. It is easy to verify that the $\boxplus$ addition is isometric, hence $(\mathcal{C}, \boxplus)$ is a linear code. However, $X_1 \cap X_2, X_1 + X_2 \notin \mathcal{C}$ and $(X_1 \cap X_2) \cap (X_1 \boxplus X_2) = Z \ne \{0\}$. For two codewords $x, y$ in a  binary linear code $\mathcal{S}$, $(x \ast y) \ast (x + y) = \mathbf{0}$ irrespective of whether $x \ast y, x \circ y \in \mathcal{S}$ or not. This is in accordance with the fact that linearity is inherent in the entirety of a Hamming space. The same does not hold for projective spaces.
\end{remark}

The Union-Intersection theorem helps us to bring out the lattice structure in a certain class of linear codes. To elaborate, we recall the definition of \emph{linear codes closed under intersection} introduced in \cite{BK}.
\begin{definition}
	\label{8}
	A linear code $\mathcal{U} \subseteq \mathbb{P}_q(n)$ with the property that $X \cap Y \in \mathcal{U}$ whenever $X, Y \in \mathcal{U}$ is said to be a \emph{linear code closed under intersection}.
\end{definition}
According to Theorem~\ref{6p} a linear code is closed under intersection if and only if it is also closed under the usual vector space addition. Since linear subspace codes are subsets of the associated projective lattice, the following statement follows as a direct consequence of Theorem~\ref{6p}.
\begin{corollary}
	\label{SL}
	Let $\mathcal{U} \subseteq \mathbb{P}_q(n)$ be a linear subspace code.	$\mathcal{U}$ is a sublattice of the projective lattice $(\mathbb{P}_q(n), +, \cap)$ if and only if $\mathcal{U}$ is closed under intersection.
\end{corollary}
\section{Pairwise Disjoint Codewords in Linear Subspace Codes}
\label{S4}

Two vectors in a Hamming space are \emph{disjoint} if their intersection is empty. It is easy to verify that a set of pairwise disjoint vectors in $\mathbb{F}_2^n$ are linearly independent over $\mathbb{F}_2$. We will prove an analogous result for linear codes in a projective space. First we formally define a set of \emph{pairwise disjoint} codewords in a linear code.
\begin{definition}
	\label{P}
	A set of codewords $\{X_1, \ldots, X_r\}$ in a linear subspace code $\mathcal{U} \subseteq \mathbb{P}_q(n)$ is \emph{pairwise disjoint} if
	\begin{equation*}
	X_i \cap X_j = \{0\} \ \ \ \ \ \ \ \ \ \ \text{for all} \ 1 \leq i \neq j \leq r.
	\end{equation*}
\end{definition}
We are now going to establish that any set of pairwise disjoint codewords in a linear subspace code is linearly independent. To this end, we need certain properties of any finite $m \ge 3$ number of pairwise disjoint codewords. First we prove the base case when $m = 3$.
\begin{lemma}
	\label{6m}
	If $X_1, X_2, X_3$ are pairwise disjoint nontrivial codewords in a linear subspace code $\mathcal{U}$ then, $X_i \cap (X_j \boxplus X_k) = \{ 0\}$ for distinct $i, j, k \in [3]$. Furthermore, $X_1 \boxplus X_2 \boxplus X_3 = X_1 + X_2 + X_3$, and $\dim(X_1 \boxplus X_2 \boxplus X_3) = \dim X_1 + \dim X_2 + \dim X_3$.
\end{lemma}
\begin{proof}
	To prove the first part of the lemma, it suffices to show that $X_3 \cap (X_1 \boxplus X_2) = \{0\}$. As $X_1 \cap X_3 = \{ 0\}$, by Lemma~\ref{4} $X_1 \boxplus X_3 = X_1 \oplus X_3$, which also implies $X_3 \subset (X_1 \boxplus X_3)$. Similarly, $X_3 \subset (X_2 \boxplus X_3)$. Thus we can write $X_3 \subseteq (X_1 \boxplus X_3) \cap (X_2 \boxplus X_3)$. According to Definition~\ref{L}, we have $X_1 \boxplus X_2 = (X_1 \boxplus X_3) \boxplus (X_2 \boxplus X_3)$ and using Lemma~\ref{2} yields
	\begin{equation*}
	\dim(X_1 \boxplus X_2) = \dim(X_1 \boxplus X_3) + \dim(X_2 \boxplus X_3) - 2\dim ((X_1 \boxplus X_3) \cap (X_2 \boxplus X_3)).
	\end{equation*}
	Since $X_1, X_2, X_3$ are pairwise disjoint codewords, using Lemma~\ref{4} the above equation can also be expressed as
	\begin{equation*}
	\begin{aligned}
	\dim X_1 + \dim X_2 = &(\dim X_1 + \dim X_3) + (\dim X_2 + \dim X_3) \\
	&- 2\dim((X_1 \boxplus X_3) \cap (X_2 \boxplus X_3)),
	\end{aligned}
	\end{equation*}
	which reduces to $\dim X_3 = \dim((X_1 \boxplus X_3) \cap (X_2 \boxplus X_3))$. But $X_3 \subseteq (X_1 \boxplus X_3) \cap (X_2 \boxplus X_3)$. Combining both, we get $X_3 = (X_1 \boxplus X_3) \cap (X_2 \boxplus X_3)$ and by Lemma~\ref{4} this is equivalent to
	\begin{equation}
	\label{E10}
	X_3 = (X_1 + X_3) \cap (X_2 + X_3).
	\end{equation}
	We now claim that $X_3 \cap (X_1 \boxplus X_2) = X_3 \cap (X_1 + X_2) = \{0\}$. Suppose not, then there must exist some nonzero $x_3 \in X_3$ such that $x_3 = x_1 + x_2$, with $x_1 \in X_1$ and $x_2 \in X_2$. Since the pairwise intersections of $X_1, X_2, X_3$ are trivial, neither $x_1$ nor $x_2$ belongs to $X_3$. Also, both $x_1$ and $x_2$ are nonzero. Then $x_2 = x_3 - x_1$, which means $x_2$ is in $X_1 + X_3$. Thus $x_2 \in (X_1 + X_3) \cap (X_2 + X_3)$, which contradicts \eqref{E10}. Hence $X_3 \cap (X_1 \boxplus X_2) = \{0\}$. Combining this with Lemma~\ref{4}, we can write:
	\begin{equation*}
	X_1 \boxplus X_2 \boxplus X_3 = (X_1 \boxplus X_2) + X_3 = X_1 + X_2 + X_3.
	\end{equation*}
	The rest follows from the fact that $X_1 \cap X_2 = X_3 \cap (X_1 + X_2) = \{0\}$.
\end{proof}



We are now in a position to prove the general case for any finite $m \ge 3$ number of pairwise disjoint codewords.
\begin{lemma}
	\label{6n}
	Let $\{Y_1, \ldots, Y_m\}$ be a set of pairwise disjoint nontrivial codewords in a linear subspace code $\mathcal{U}$. Then,
	\begin{itemize}
		\item[(a)] for all $j \in [m]$, $Y_j \cap \sum_{i \in [m] \backslash \left\{j\right\}}{Y_i} = \{0\}$; 
		\item[(b)] $Y_1 \boxplus Y_2 \boxplus \cdots \boxplus Y_m = Y_1 + Y_2 + \cdots + Y_m$; and
		\item[(c)] $\dim(Y_1 \boxplus Y_2 \boxplus \cdots \boxplus Y_m) = \sum_{i=1}^m \dim Y_i$.
	\end{itemize}
\end{lemma}
\begin{proof}
	We prove (a)--(c) simultaneously by induction on $m$, the number of pairwise disjoint codewords. The base case of two codewords for (b)--(c) is covered by Lemma~\ref{4} while that for (a) is because of the assumption of pairwise disjointness. As the induction hypothesis, assume that the statements (a)--(c) hold for any set of $m-1$ pairwise disjoint codewords, for some $m \ge 3$. In particular, for any $(m-1)$-subset $\mathcal{I} \subset [m]$, we have $\boxplus_{i \in \mathcal{I}} Y_i = \sum_{i \in \mathcal{I}} Y_i$ and $\dim(\boxplus_{i \in \mathcal{I}} Y_i) = \sum_{i\in\mathcal{I}}\dim Y_i$.
	
	Consider $Z = {\boxplus}_{i \in [m-2]} Y_i$. By the induction hypothesis, $Z = {\sum}_{i=1}^{m-2} Y_i$ and $\dim Z = {\sum}_{i=1}^{m-2} \dim Y_i$, whereas $Z \cap Y_{m-1} = Z \cap Y_m = \{0\}$. Thus, according to Definition~\ref{P} $Z, Y_{m-1}, Y_m$ are pairwise disjoint nontrivial codewords of $\mathcal{U}$ and by Lemma~\ref{6m}, $(Z \boxplus Y_{m-1}) \cap Y_m = (Y_1 \boxplus \cdots \boxplus Y_{m-1}) \cap Y_m = \{0\}$. Hence, (a) must hold for $Y_1, \ldots, Y_m$.
	
	To prove (b), observe that $Y_m \cap \sum\limits_{i=1}^{m-1} Y_i = \{0\}$ according to part (a). Then, by the induction hypothesis and Lemma~\ref{4} we have:
	\begin{equation*}
	\boxplus_{i=1}^{m} Y_i = (\boxplus_{i=1}^{m-1} Y_i) \boxplus Y_m =  \left(\sum\limits_{i=1}^{m-1} Y_i\right) \boxplus Y_m = \left(\sum\limits_{i=1}^{m-1} Y_i\right) + Y_m = \sum\limits_{i=1}^{m} Y_i.
	\end{equation*}
	
	Finally, parts (a), (b) and Lemma~\ref{4} imply:
	\begin{eqnarray*}
		\dim(\boxplus_{i=1}^{m} Y_i) &=& \dim((\boxplus_{i=1}^{m-1} Y_i) \boxplus Y_m) 
		= \dim\left(\left(\sum\limits_{i=1}^{m-1} Y_i\right) \boxplus Y_m\right) \\
		&=& \dim\left(\sum\limits_{i=1}^{m-1} Y_i\right) + \dim Y_m \\
		&=& \left(\sum\limits_{i=1}^{m-1} \dim Y_i\right) + \dim Y_m = \sum\limits_{i=1}^{m} \dim Y_i,
	\end{eqnarray*}
	which proves part(c).
\end{proof}
\begin{theorem}
	\label{IB}
	Any set of pairwise disjoint codewords in a linear subspace code is linearly independent over $\mathbb{F}_2$ with respect to the corresponding linear addition.
\end{theorem}
\begin{proof}
	Suppose that $\{X_1, \ldots, X_m\}$ is a set of pairwise disjoint codewords in a linear code $\mathcal{U}$ with corresponding linear addition $\boxplus$. The statement is trivially true for $m = 2$. To prove it for $m \ge 3$, assume the contrary, i.e. there exists a minimal positive integer $2 \le r \le m$ such that $r$ of the $m$ indecomposable codewords are linearly dependent. Thus, there exist positive integers $i_1, \ldots, i_r$, where $1 \le i_1 < \cdots < i_r \le m$, such that
	\begin{equation*}
	X_{i_1} \boxplus \cdots \boxplus X_{i_r} = \{0\}.
	\end{equation*}
	According to Lemma~\ref{6n}(b), the above equation reduces to
	$\sum_{j=1}^{r} X_{i_j} = \{0\}$, which is a contradiction as each of $X_{i_1}, \ldots, X_{i_r}$ is nontrivial and the result follows.
\end{proof}

\section{Lattice Structure of Linear Codes in Projective Spaces}
\label{S5}

	The maximum possible size of a linear code closed under intersection in $\mathbb{P}_q(n)$ was proved to be $2^n$ in \cite{BK}. We record the formal statement below.
	\begin{theorem}(\cite{BK}, Proposition 17)
		\label{UB}
		If $\mathcal{U}$ is a linear code in $\mathbb{P}_q(n)$ that is closed under intersection then $|\mathcal{U}| \le 2^n$.
	\end{theorem}

	The proof of the above theorem relies on the notion of \emph{indecomposable codewords}, which were introduced in \cite{BK}. We record the formal definition here along with a few important properties of indecomposable codewords that were proved in \cite{BK}. The first two proofs are omitted.
	\begin{definition}
		\label{9}
		A codeword $Y \neq \{0\}$ of a linear subspace code $\mathcal{U}$ is said to be \emph{indecomposable} if $Y$ cannot be expressed as $Y = Y_1 \boxplus Y_2$ for any $Y_1,Y_2 \in \mathcal{U}$ with $\dim Y_1, \dim Y_2 < \dim Y$.
	\end{definition}
	\begin{remark}
		Note that $\{0\}$ can never be an indecomposable codeword in a linear code in a projective space.
	\end{remark}
	\begin{lemma}(\cite{BK}, Lemma~9)
	\label{lem:basic}
		Let $Y$ be an indecomposable codeword of a linear subspace code $\mathcal{U}$. Then, for any codeword $X \in \mathcal{U}$, we have $X \subseteq Y$ iff $X = \{0\}$ or $X=Y$.
	\end{lemma}
\begin{lemma}(\cite{BK}, Lemma~10)
	\label{10}
	If $Y_1, Y_2$ are any two distinct indecomposable codewords of a linear subspace code $\mathcal{U}$ that is closed under intersection, then $Y_1 \cap Y_2 = \left\{0\right\}$. Consequently, $Y_1 \boxplus Y_2 = Y_1 + Y_2$, and $\dim(Y_1 \boxplus Y_2) = \dim Y_1 + \dim Y_2$.
\end{lemma}

Thus, the collection of indecomposable codewords in a linear code closed under intersection is pairwise disjoint. The following result is then a direct consequence of Lemma~\ref{6n}(b).
\begin{lemma}(\cite{BK}, Lemma~11(b))
	\label{lem:gen}
	Let $Y_1,Y_2,\ldots,Y_m$, $m \ge 2$, be distinct indecomposable codewords of a linear subspace code $\mathcal{U}$ that is closed under intersection. Then,
	\begin{equation*}
		 Y_1 \boxplus Y_2 \boxplus \cdots \boxplus Y_m = Y_1 + Y_2 + \cdots + Y_m.
	\end{equation*}
\end{lemma}

It was established in \cite{BK} that any codeword of a linear code closed under intersection can be decomposed into a $\boxplus$-sum of its indecomposable codewords. Also, such a decomposition is unique.
\begin{proposition}(\cite{BK}, Proposition~12)
	\label{11}
	Let $\mathcal{U}$ be a linear subspace code that is closed under intersection, and let $Y_1, Y_2, \ldots, Y_m$ be its indecomposable codewords. Then, any codeword $X \in \mathcal{U}$ can be uniquely expressed as $\boxplus_{i \in \mathcal{I}} Y_i$ for some $\mathcal{I} \subseteq [m]$.
\end{proposition}

We now use the notion of indecomposable codewords to bring out an important property of the sublattice of the projective lattice formed by a linear subspace code closed under intersection, namely, that the sublattice is geometric distributive.
\begin{theorem}
	The sublattice of the projective lattice formed by a linear code closed under intersection is geometric distributive.
\end{theorem}
\begin{proof}
	Let $\mathcal{U}$ be a linear code in $\mathbb{P}_q(n)$ closed under intersection. Since the projective lattice $(\mathbb{P}_q(n), +, \cap)$ is geometric, so is any sublattice of it. By Corollary~\ref{SL}, $\mathcal{U}$ is a sublattice of $\mathbb{P}_q(n)$, thus it suffices to show that $\mathcal{U}$ is distributive. In particular, we need to prove that $X_1 \cap (X_2 + X_3) = (X_1 \cap X_2) + (X_1 \cap X_3)$ for all $X_1, X_2, X_3 \in \mathcal{U}$ (see Example~\ref{PL}). Suppose $\{Y_1, \ldots, Y_m\}$ is the set of all indecomposable codewords in $\mathcal{U}$. Proposition~\ref{11} then allows us to write:
	\begin{equation*}
	X_1 = \boxplus_{i \in \mathcal{I}_1} Y_i, \ \ X_2 = \boxplus_{j \in \mathcal{I}_2} Y_j, \ \ X_3 = \boxplus_{l \in \mathcal{I}_3} Y_l,
	\end{equation*}
	for fixed $\mathcal{I}_1, \mathcal{I}_2, \mathcal{I}_3 \subseteq [m]$. Observe that for $A, B \in \mathcal{U}$ such that $A = \boxplus_{i \in S_A} Y_i$ and $B = \boxplus_{j \in S_B} Y_j$ where $S_A, S_B \subseteq [m]$, we have $A \boxplus B = \boxplus_{i \in S_A \triangle S_B} Y_i$, and by Lemma~\ref{lem:gen},
	\begin{eqnarray}
	\label{C1}
	A+B &=& \left(\boxplus_{i \in S_A} Y_i\right) + \left(\boxplus_{j \in S_B} Y_j\right) \nonumber \\
	&=& \biggl(\sum_{i \in S_A} Y_i\biggr) + \biggl(\sum_{j \in S_B} Y_j\biggr) = \sum_{i \in S_A \cup S_B} Y_i.
	\end{eqnarray}
	Theorem~\ref{6p} and Lemma~\ref{lem:gen} then together imply that,
	\begin{eqnarray}
	\label{C2}
	A \cap B &=& A \boxplus B \boxplus (A+B) \nonumber \\
	&=& (\boxplus_{i \in S_A} Y_i) \boxplus (\boxplus_{j \in S_B} Y_j) \boxplus \left(\sum_{l \in S_A \cup S_B} Y_l\right) \nonumber \\
	&=& (\boxplus_{i \in S_A} Y_i) \boxplus (\boxplus_{j \in S_B} Y_j) \boxplus (\boxplus_{l \in S_A \cup S_B} Y_l) \nonumber \\
	&=& \boxplus_{i \in S_A \cap S_B} Y_i = \sum_{i \in S_A \cap S_B} Y_i.
	\end{eqnarray}
	
	We must have $X_1 \cap (X_2 + X_3), (X_1 \cap X_2) + (X_1 \cap X_3) \in \mathcal{U}$ by Theorem~\ref{6p}. According to \eqref{C1} and \eqref{C2} we then compute the following:
	\begin{eqnarray}
	\label{C3}
		X_1 \cap (X_2 + X_3) &=& \sum\limits_{i \in \mathcal{I}_1 \cap (\mathcal{I}_2 \cup \mathcal{I}_3)} Y_i; \\
		(X_1 \cap X_2) + (X_1 \cap X_3) &=& \left(\sum\limits_{j \in \mathcal{I}_1 \cap \mathcal{I}_2} Y_j\right) + \left(\sum\limits_{j \in \mathcal{I}_1 \cap \mathcal{I}_3} Y_l\right) \nonumber \\ 
		&=& \sum\limits_{j \in (\mathcal{I}_1 \cap \mathcal{I}_2) \cup (\mathcal{I}_1 \cap \mathcal{I}_3)} Y_j. \label{C4}
	\end{eqnarray}

	For subsets $\mathcal{I}_1, \mathcal{I}_2, \mathcal{I}_3 \subseteq [m]$ we have, by distributivity of set intersection over set union,
	\begin{equation*}
	\mathcal{I}_1 \cap (\mathcal{I}_2 \cup \mathcal{I}_3) = (\mathcal{I}_1 \cap \mathcal{I}_2) \cup (\mathcal{I}_1 \cap \mathcal{I}_3),
	\end{equation*}
	which along with \eqref{C3} and \eqref{C4} prove that $X_1 \cap (X_2 + X_3) = (X_1 \cap X_2) + (X_1 \cap X_3)$. Hence $\mathcal{U}$ is distributive.
\end{proof}
	\begin{remark}
	In the preceding proof, $\mathcal{U}$ is a geometric sublattice of the projective lattice. By definition, any nontrivial codeword $Z \in \mathcal{U}$ is a join of atoms. We can deduce that the atoms in $\mathcal{U}$ are precisely the set of its indecomposable codewords: That an indecomposable codeword in $\mathcal{U}$ is an atom follows directly from Lemma~\ref{lem:basic}. To prove the converse, we express an atom $X \in \mathcal{U}$ in the following way using Proposition~\ref{11} and Lemma~\ref{lem:gen}: $$X = \boxplus_{i \in \mathcal{I}} Y_i = \sum_{i \in \mathcal{I}} Y_i$$ for some nonempty $\mathcal{I} \subseteq [m]$. Clearly $\{0\} < Y_i < X$ for any $i \in \mathcal{I}$, which indicates that $|\mathcal{I}| = 1$, i.e., $X = Y_i$ for some $i \in [m]$.
\end{remark}
\begin{remark}
	\label{R}
	The set of indecomposable codewords in a linear subspace code closed under intersection is linearly independent with respect to the linear addition over $\mathbb{F}_2$. This together with Proposition~\ref{11} imply that the indecomposable codewords are a basis for the vector space over $\mathbb{F}_2$ formed by the linear code.
\end{remark}

As an application of the lattice-theoretic characterization of the linear subspace codes closed under intersection, we now give an alternative proof of the upper bound on size of such class of codes. First we need to state a result that follows directly from the \emph{fundamental theorem of finite distributive lattices} by Birkhoff \cite{B}.
\begin{theorem}(\cite{B}, Ch.~IX , Sec.~4, Ex.~1)
	\label{FT}
	A distributive lattice of height $n$ contains at most $2^n$ elements.
\end{theorem}
\begin{proof}(Proof of Theorem~\ref{UB})
	The height of the projective lattice $\mathbb{P}_q(n)$ is $n$. Thus $\mathcal{U}$, a distributive sublattice of $\mathbb{P}_q(n)$, is of height at most $n$, whence $|\mathcal{U}| \le 2^n$ follows using Theorem~\ref{FT}.
\end{proof}
\begin{remark}
	The number of indecomposable codewords that a linear code closed under intersection in $\mathbb{P}_q(n)$ admits can be at most $n$. However, there are examples of linear codes in $\mathbb{P}_q(n)$ that are not closed under intersection, in which the number of indecomposable codewords is as high as $(2^n - 1)$ (E.g. \cite{BEV}, Example~1).
\end{remark}
\section{Conclusion}
\label{S6}

We have studied similarities in the structure of binary linear block codes and linear codes in a projective space; and explored lattice structure in linear subspace codes. Our findings indicate that the only class of linear codes in $\mathbb{P}_q(n)$ that are sublattices of the projective lattice $(\mathbb{P}_q(n), + , \cap)$ are those closed under intersection. Such linear codes were shown to have maximum size of at most $2^n$ \cite{BK}. However there are examples of other linear codes which are not closed under intersection yet attain the bound conjectured by Braun et al. These class of codes cannot be studied within a lattice framework. Thus it is necessary to view linear codes in a projective space using a more general setting. Whether the BEV bound (Conjecture~\ref{C}) holds true for all linear codes remains an interesting open problem.

We observed earlier that the indecomposable codewords in a linear subspace code constitute a basis for the vector space over $\mathbb{F}_2$ formed by the code (Remark~\ref{R}). We refer to such a basis as an \emph{indecomposable basis}. A linear code closed under intersection has a unique indecomposable basis. However, there are examples of linear codes in projective spaces that have more indecomposable codewords than the dimension of the codes (\cite{BEV}, Example~1, Remark~3). Such linear codes may not possess a unique indecomposable basis. In fact we conjecture the following.

\begin{conjecture}
	\label{UIB}
	A linear code $\mathcal{U}$ in $\mathbb{P}_q(n)$ has a unique indecomposable basis if and only if $\mathcal{U}$ is closed under intersection.
\end{conjecture}



\bibliography{mybibfile}

\begin{thebibliography}{10}
\expandafter\ifx\csname url\endcsname\relax
  \def\url#1{\texttt{#1}}\fi
\expandafter\ifx\csname urlprefix\endcsname\relax\def\urlprefix{URL }\fi
\expandafter\ifx\csname href\endcsname\relax
  \def\href#1#2{#2} \def\path#1{#1}\fi

\bibitem{KK}
R.~Koetter, F.~R. Kschischang, Coding for errors and erasures in random network
  coding, IEEE Transactions on Information Theory 54~(8) (2008) 3579--3591.

\bibitem{AAK}
R.~Ahlswede, H.~K. Aydinian, L.~H. Khachatrian, On perfect codes and related
  concepts, Designs, Codes and Cryptography 22~(3) (2001) 221--237.

\bibitem{EV}
T.~Etzion, A.~Vardy, Error-correcting codes in projective space, IEEE
  Transactions on Information Theory 57~(2) (2011) 1165--1173.

\bibitem{SE}
N.~Silberstein, T.~Etzion, Enumerative coding for grassmannian space, IEEE
  Transactions on Information Theory 57~(1) (2010) 365--374.

\bibitem{SE2}
N.~Silberstein, T.~Etzion, Large constant dimension codes and lexicodes,
  Advances in Mathematics of Communications 5~(2) (2011) 177--189.

\bibitem{HKK}
T.~Honold, M.~Kiermaier, S.~Kurz, Johnson type bounds for mixed dimension
  subspace codes, arXiv preprint arXiv:1808.03580 (2018).

\bibitem{BEV}
M.~Braun, T.~Etzion, A.~Vardy, Linearity and complements in projective space,
  Linear Algebra and its Applications 438~(1) (2013) 57--70.

\bibitem{MS}
F.~J. MacWilliams, N.~J.~A. Sloane, The theory of error-correcting codes,
  Vol.~16, Elsevier, 1977.

\bibitem{DR}
D.~Ruano, The metric structure of linear codes, in: Singularities, Algebraic
  Geometry, Commutative Algebra, and Related Topics, Springer, 2018, pp.
  537--561.

\bibitem{GY}
M.~Gadouleau, Z.~Yan, Constant-rank codes and their connection to
  constant-dimension codes, IEEE Transactions on Information Theory 56~(7)
  (2010) 3207--3216.

\bibitem{XF}
S.-T. Xia, F.-W. Fu, Johnson type bounds on constant dimension codes, Designs,
  Codes and Cryptography 50~(2) (2009) 163--172.

\bibitem{TMBR}
A.-L. Trautmann, F.~Manganiello, M.~Braun, J.~Rosenthal, Cyclic orbit codes,
  IEEE Transactions on Information Theory 59~(11) (2013) 7386--7404.

\bibitem{KoK}
A.~Kohnert, S.~Kurz, Construction of large constant dimension codes with a
  prescribed minimum distance, in: Mathematical methods in computer science,
  Springer, 2008, pp. 31--42.

\bibitem{ER}
T.~Etzion, N.~Raviv, Equidistant codes in the grassmannian, Discrete Applied
  Mathematics 186 (2015) 87--97.

\bibitem{BP}
D.~Bartoli, F.~Pavese, A note on equidistant subspace codes, Discrete Applied
  Mathematics 198 (2016) 291--296.

\bibitem{GR}
E.~Gorla, A.~Ravagnani, Equidistant subspace codes, Linear Algebra and its
  Applications 490 (2016) 48--65.

\bibitem{B}
G.~Birkhoff, Lattice theory, Vol.~25, American Mathematical Soc., 1940.

\bibitem{BK}
P.~Basu, N.~Kashyap, On linear subspace codes closed under intersection, in:
  2015 Twenty First National Conference on Communications (NCC), IEEE, 2015,
  pp. 1--6.

\bibitem{PS}
B.~S. Pai, B.~S. Rajan, On the bounds of certain maximal linear codes in a
  projective space, IEEE Transactions on Information Theory 61~(9) (2015)
  4923--4927.

\end{thebibliography}

\end{document}